\documentclass[twoside,12pt]{article}
\usepackage{amsmath,amssymb,amsthm,amsfonts,mathrsfs,wasysym,latexsym,times,lineno, subfigure,color,booktabs}
\usepackage{amsmath,amsfonts}
\usepackage{multirow}
\usepackage{array}
\usepackage{amsthm}
\usepackage{graphicx}
\usepackage{color}
\usepackage{multicol}
\usepackage{float}
\usepackage{cite}

\newtheorem{thm}{Theorem}[section]

\newtheorem{definition}{Definition}[section]

\newtheorem{lemma}[thm]{Lemma}
\newtheorem{remark}{Remark}[section]
\newtheorem{theorem}[thm]{Theorem}

\newtheorem{corollary}[thm]{Corollary}

  \date{}
\begin{document}
\pagenumbering{roman}
\begin{center}
\fontsize{20pt}{20pt}\selectfont{Construction Is All You Need}\\
\vspace{0.5em}
\fontsize{16pt}{20pt}\selectfont{------Extended Abstract of ``SAT Requires Exhaustive Search" } 
\end{center}
\vspace{1em}
\qquad Our paper introduces a new framework for studying computational hardness, which can be seen as an extension of Gödel’s framework for proving the incompleteness theorem. Gödel’s framework demonstrates that a constructed self-referential proposition is unprovable, meaning that reasoning based on syntax cannot determine the semantics of this proposition. Within this new framework, all aspects of computational hardness, including the P vs NP problem, should be reconsidered. This famous problem obscures the surprising fact that we do not know any partial result of complexity lower bounds (e.g., a super-linear lower bound), but we know barriers to obtaining such results. In many cases, barriers arise simply because we are not on the right track. We need to explore new approaches for proving lower bounds from the starting point of computer science. The goal is to understand the limits of computation and clarify why some problems are harder than others.

Computation is fundamentally a mechanical process of reasoning. In other words, it can be viewed as mechanized mathematics. There are two main lines of research on the limits of computation. Gödel’s results address the limits of reasoning within finite formal systems, while Turing’s results explore the limits of whether mechanical processes can be realized in the physical universe. Computational time fundamentally relies on the reasoning of concrete examples rather than the computability of abstract problems. The P vs NP problem is a natural extension of Turing’s concept of computability. However, even if a problem is uncomputable, each instance might still be solved quickly, though no general (mechanical) method exists to handle all instances. In other words, computability is unrelated to computational time. Therefore, Turing’s framework of computability is not suitable for studying computational hardness. This also explains why there are so many barriers in current complexity theory studies: we might have taken the wrong path from the very beginning.   

In fact, the Turing machine itself is an assumption about machines (i.e., mechanized finite formal systems) that can be realized in the physical universe. In essence, the Turing machine represents a fundamental physical assumption, and Turing's findings on uncomputability signify the physical limits of mechanically solving all instances of a problem. The computational hardness of concrete examples arises from the mathematical limits of reasoning. It is unnecessary to discuss the mechanical solving of these examples based on the concept of Turing machines. This is because if the reasoning of concrete examples requires a long time, then solving these examples mechanically will also take a long time. We can also say that computational hardness results are essentially mathematical impossibility results of reasoning. Any mathematical impossibility results cannot be proved without restrictions on the form of mathematics. To derive such results, it is essential to make a mathematical (rather than physical) assumption about reasoning.

The essence of computational hardness lies not in the contradiction between deterministic polynomial computation and non-deterministic polynomial computation, but in the contradiction between non-brute-force computation and brute-force computation. This contradiction essentially questions whether syntax can replace semantics (i.e., whether the part can replace the whole through reasoning). For CSPs, this question reduces to whether a certain number of subproblems can replace the whole problem, based on a mathematical assumption about the reasoning of CSPs. We then construct self-referential examples of CSP and ultimately prove by contradiction that a certain number of subproblems cannot replace the whole problem for these CSP examples. This is very similar to Gödel’s study on whether a finite set of true mathematical statements (i.e., axioms) can replace an infinite set of all true mathematical statements through reasoning. Therefore, our result can also be regarded as an incompleteness theorem for CSPs.

Gödel’s proof and our proof address two different cases (i.e., first-order logic and propositional logic, respectively). The essence of both proofs (i.e., the use of self-reference) is the same, but there are differences in constructing the self-referential object. Specifically, this involves constructing an object such that negating it results in an object equivalent to itself. In Gödel’s work, the self-referential object is a logical formula. In our work, the self-referential object is an infinite set of satisfiable and unsatisfiable examples. These self-referential examples necessitate brute-force computation because the part cannot replace the whole for these examples. In other words, reasoning based on syntax is ineffective, and only brute-force computation based on semantics can solve these examples. It should also be emphasized that our hardness results apply to finite formal systems under a mathematical assumption about reasoning. Turing machines, on the other hand, are finite formal systems under a physical assumption about machines. Thus, while both are subsets of a larger set, they cannot be directly compared. In other words, our work and current complexity theory can be seen as two different paths originating from the same point, both aiming to explain why some problems are harder than others.  

Finally, we address the two main barriers in current complexity theory. The first, the relativization barrier, indicates that the self-reference (diagonalization) method cannot separate P from NP. In our work, the self-reference method is applied to constructed satisfiable and unsatisfiable examples, thus avoiding the relativization barrier. Regarding the natural proof barrier, the self-reference method can also circumvent it. Therefore, employing the self-reference method on constructed examples allows us to bypass both barriers. Many people believe that complexity theory is difficult because it deals with computational hardness. This is a misunderstanding. On the contrary, our work shows that extreme hardness (i.e., brute-force computation) is easy to understand through the construction of self-referential examples. In other words, construction is all you need.
\newpage
\pagenumbering{arabic}

\title{SAT Requires Exhaustive Search}
\author{Ke Xu$^{1}$,\quad Guangyan Zhou$^{2}$\\
\footnotesize $^{1}$State Key Lab of Software Development Environment, Beihang University, Beijing, 100083, China \\
\footnotesize kexu@buaa.edu.cn
\\
\footnotesize $^2$Department of Mathematics and Statistics, Beijing Technology and Business University, Beijing, 100048, China\\
\footnotesize zhouguangyan@btbu.edu.cn
}

\maketitle

\begin{abstract}

In this paper, by constructing extremely hard examples of CSP (with large domains) and SAT (with long clauses), we prove that such examples cannot be solved without exhaustive search, which is stronger than P $\neq$ NP. This constructive approach for proving impossibility results is very different (and missing) from those currently used in computational complexity theory, but is similar to that used by Kurt G\"{o}del in proving his famous logical impossibility results. Just as shown by G\"{o}del's results that proving formal unprovability is feasible in mathematics, the results of this paper show that proving computational hardness is not hard in mathematics. Specifically, proving lower bounds for many problems, such as 3-SAT, can be challenging because these problems have various effective strategies available for avoiding exhaustive search. However, in cases of extremely hard examples, exhaustive search may be the only viable option, and proving its necessity becomes more straightforward. Consequently, it makes the separation between SAT (with long clauses) and 3-SAT much easier than that between 3-SAT and 2-SAT. Finally, the main results of this paper demonstrate that the fundamental difference between the syntax and the semantics revealed by G\"{o}del's results also exists in CSP and SAT.

\end{abstract}

\section{Introduction}
Model RB is a random constraint satisfaction problem (CSP) model that was proposed by Xu and Li \cite{xu2000} in 2000, which could also be encoded to well-known NP-complete problems like SAT and CLIQUE. The purpose of this model was to address the issue of trivial unsatisfiability that was prevalent in previous random CSP models.  One of the key features of Model RB is that its domain size $d$ grows with the number of variables $n$. Additionally, Model RB has been proved to exhibit exact phase transitions from satisfiability to unsatisfiability, making it a useful tool for analyzing and evaluating the performance of algorithms. Over the last two decades, Model RB has been extensively researched from multiple perspectives, as evidenced by various studies (e.g., \cite{xu2005,xu2006,xuAI2007,liu2011,cai2011,ZZ2011,saitta2011,zhao2012,fan2012,lecoutre2013,huang2014,xu2015,liu2015,knuth2015,xuwei2016,fang2016,wang2016,li2018,kara2020,zhou2022}).
Moreover, this model has gained significant popularity and widespread use in renowned international algorithm competitions. A random instance of Model RB with a planted solution named frb100-40, where $n=100$ and $d=40$, has remained elusive since it was made available online in 2005 as a 20-year challenge for algorithms$\footnote{See https://tinyurl.com/2p53xbd7}$. Despite numerous attempts, no one has been able to solve it thus far. In summary, the results suggest that Model RB possesses nice mathematical properties that can be easily derived.  In contrast to its mathematical tractability, the random instances of this model, particularly those generated in the phase transition region, present a significant challenge for various algorithms, proving to be extremely difficult to solve.

As shown in the proof of G\"{o}del's well-known incompleteness theorem \cite{godel1931}, the constructive approach plays an indispensable role in revealing the fundamental limitations of finite formal systems. An algorithm can be regarded as a finite formal system to deduce conclusions step by step. In this paper, we will study the limitations of algorithms based on the constructive approach. Specifically, we will investigate whether it is always possible to develop an algorithm that can replace the naive exponential exhaustive search. The Strong Exponential Time Hypothesis \cite{calabro2009} conjectures that no such algorithm exists for SAT. This problem is similar to that proposed by David Hilbert in the early period of 20th century if it is always possible to construct a finite formal system that can replace an area of mathematics (arithmetic for example) containing infinitely many true mathematical statements. These two problems ask essentially the same philosophical question: \emph{whether the part can always replace the whole within the limits of reasoning}. We think that it is more fundamental and has a higher priority than whether an efficient (polynomial-time) algorithm exists.

The advantages of Model RB enable us to choose specific threshold points at which instances with a symmetry requirement  are on the edge of being satisfiable and unsatisfiable. In fact, we will show that there exist instances at exactly the same point which are either satisfiable with exactly one solution or unsatisfiable but only fail on one constraint. The satisfiability of such instances can be flipped under a special symmetry mapping, thus the set of these instances form a fixed point under the symmetry mapping and this enables us to construct the most indistinguishable examples that can be understood as the source of computational hardness.  Based on the symmetry mapping and driven by the famous method of diagonalization and self-reference, we show that unless exhaustive search is executed, the satisfiability of a certain constraint (thus the whole instance) is possible to be changed, while the subproblems of the whole instance remain unchanged. Therefore, whether the whole instance is satisfiable or unsatisfiable cannot be distinguished without exhaustive search. In summary, if we can construct the most indistinguishable examples with exactly the same method and the same parameter values (which is very special and not an easy task), then it is not hard to understand and prove why they are extremely hard to solve.
\section{Model RB}

 A random instance $I$ of Model RB consists of the following:
 \begin{itemize}
   \item  A set of variables $\mathcal{X}=\{x_1,...,x_n\}$: Each
   variable $x_i$ takes values from its domain $D_i$, and the domain size is $|D_i|=d$, where $d=n^\alpha$ for $i=1,...,n$, and $\alpha>0$ is a constant.
   \item  A set of constraints $\mathcal{C}=\{C_1,...,C_m\}$ ($m=rn\ln d$, where $r>0$ is a constant): for each
   $i=1,...,m$, constraint $C_i=(X_i,R_i)$. $X_i=(x_{i_1},x_{i_2},...,x_{i_k})$
   ($k\geq2$ is a constant) is a sequence of $k$ distinct variables chosen
   uniformly at random without repetition from $\mathcal{X}$. $R_i$
   is the permitted set of tuples of values which are selected uniformly without repetition
   from the subsets of $ D_{i_1}\times
D_{i_2}\times\cdots\times D_{i_k}$, and $|R_i|=(1-p)d^k$ where $0<p<1$
is a constant.
 \end{itemize}

 In this paper, we have a symmetry requirement of the permitted set of each constraint, and the $m$ permitted sets will be generated in the following way. Initially,  we generate a symmetry set $R$ which contains $(1-p)d^k$ tuples of values, then generate each permitted set $R_i$ of the constraint $C_i$ ($i=1,2,...,m$) by running random permutations of  domains of $k-1$ variables in $X_i$ based on $R$. For example, if $k=2$ and the domains are $D_1=D_2=\{1,2,3,4\}$, then $R=\{(1,1),(1,2),(2,1),(2,2),(3,3),(3,4),(4,3),(4,4)\}$ is a symmetry set. If we run a random permutation of $D_1$, e.g., $f(1)=3,f(2)=1,f(3)=4,f(4)=2$, then we get a permitted set $\{(3,1),(3,2),(1,1),(1,2),(4,3),(4,4),(2,3),(2,4)\}$. Through this method all $R_i(i=1,...,m)$ are isomorphic and every domain value of the variables share the same properties.

A constraint $C_i=(X_i,R_i)$ is said to be satisfied by an assignment $\sigma\in D_1\times
D_2\times\cdots\times D_n$ if the values assigned to $X_i$ are in the set $R_i$. An assignment $\sigma$ is called a solution if it satisfies all the constraints. $I$ is called satisfiable if there exists a solution, and called unsatisfiable if there is no solution.
 It has been proved that Model RB not only avoids the trivial asymptotic behavior but also has exact phase transitions from satisfiability to unsatisfiability.
Indeed, denote $\mathbf{Pr}[I\text{ is SAT}]$ the probability that a random instance $I$ of Model RB is satisfiable, then

\begin{thm}[\cite{xu2000}]
Let $r_{cr} = \frac{1}{-\ln(1-p)}$.
If $\alpha>1/k$, $0<p<1$ are two constants and $k,p$ satisfy the inequality $k\ge1/(1-p)$, then
\begin{align*}
\lim_{n\rightarrow \infty}\mathbf{Pr}[I\text{ is SAT }]=1\ \text{ if }\ r < r_{cr},\\
\lim_{n\rightarrow \infty}\mathbf{Pr}[I\text{ is SAT }]=0\ \text{ if }\ r > r_{cr}.
\end{align*}
\end{thm}

In the following we will present some properties of Model RB which are important to prove our main theorems in the next section. From here on we tacitly take $r=r_{cr}+\frac{\delta}{n\ln d}$, where $\delta=\frac{\ln2}{-\ln(1-p)}$, and take
\begin{align}\label{de:alpha}
\alpha>\max\left\{1,\inf\{\alpha:\omega<0\},-2\big(\frac{100}{99}\big)^2\frac{\ln(1-p)}{k},\frac{100\ln(1-p)}{k\ln(1-\frac p3)}\right\},
\end{align}
  where $\omega=1+\alpha(1-r_{cr}pk)$. In fact note that $\omega=1+\alpha\left(1+\frac{pk}{\ln(1-p)}\right),$ and a simple calculation yields that $pk+\ln(1-p)>0$ for all $p\in(0,1)$ under the condition that $k\ge\frac1{1-p}$. Thus it is possible to take $\alpha>0$  large enough such that $\omega<0$.

 First,  we bound the probability that a random RB instance is satisfiable.

\begin{lemma}\label{lem:onethird}
Let $I$ be a random CSP instance of Model RB with $n$ variables and $rn\ln d$ constraints. Then
$$\frac13\le \mathbf{Pr}(I\text{ is SAT})\le\frac12.$$
\end{lemma}
\begin{proof}
Let $X$ be the number of solutions of $I$, then
\begin{equation}
\mathbf{Pr}(X>0)\le \mathbf{E}[X]=d^n(1-p)^{rn\ln d}=\frac12.
\end{equation}
 As shown in \cite{xu2000},
\begin{align}\label{eq:secondmm}
\mathbf{E}[X^2]&=\sum_{S=0}^nd^n\binom n S (d-1)^{n-S}\left((1-p)\frac{\binom Sk}{\binom nk}+(1-p)^2(1-\frac{\binom Sk}{\binom nk})\right)^{rn\ln d}\\
\nonumber&= \mathbf{E}[X]^2\left(1+O(\frac1n)\right)\sum_{S=0}^nF(S),
\end{align}
where $F(S)=\binom n S (1-\frac1d)^{n-S}(\frac1d)^S\left[1+\frac{p}{1-p}s^k\right]^{rn\ln d}$,  and $S=ns$ is the number of variables for which an assignment pair take the same values.
Note that $\alpha>1$, using an argument similar to that in \cite{xu2000,ZZ2011}, we obtain that only the terms near $S=0$ and $S=n$ are not negligible, and
\begin{equation}\label{eq:cauchy}
\frac{\mathbf{E}[X^2]}{\mathbf{E}[X]^2}\le1+\frac1{\mathbf{E}[X]}+o(1)=3+o(1).
\end{equation}
Indeed, asymptotic calculations show that
\begin{align*}
&F(0)=1-o(1),F(i)=(1+o(1))n^{i(1-\alpha)},...,\\
&F(n-i)=(1+o(1))\exp\{i(\ln n+\ln d-pkr\ln d)\}/\mathbf{E}[X],F(n)=1/\mathbf{E}[X],
\end{align*}
where $i=1,2,...$ is an integer. Note that $\exp\{i(\ln n+\ln d-pkr\ln d)\}=n^{i\omega+o(1)}$ and $\alpha>1,\omega<0$ are constants, thus the upper bound of (\ref{eq:cauchy}) comes from $F(0)$ and $F(n)$.

Using the Cauchy inequality, we get $\mathbf{Pr}(X>0)\ge\mathbf{E}[X]^2/\mathbf{E}{[X^2]}\ge\frac13$.
\end{proof}
As an immediate consequence of Lemma \ref{lem:onethird} we obtain a lower bound of the probability that $I$ has exactly one solution.
 \begin{corollary}\label{lem:solutionpair}
Let $I$ be a random CSP instance of Model RB with $n$ variables and $rn\ln d$ constraints. Then the probability that $I$ has exactly one solution is at least $1/6$.
\end{corollary}
\begin{proof}
Let $\rho_1$ be the probability that $I$ has exactly one solution, and $\rho_{\ge2}$ be the probability that $I$ has at least two solutions. Then from Lemma \ref{lem:onethird}, we have
\begin{align*}
&\ \mathbf{E}[X]=\frac12\ge\rho_1+2\rho_{\ge2},\\
&\mathbf{Pr}(X>0)=\rho_1+\rho_{\ge2}\ge\frac13.
\end{align*}
Therefore $\rho_1\ge1/6$.
\end{proof}

Next we show that if a random instance is unsatisfiable, then w.h.p.$\footnote{We say a property holds w.h.p. (with high probability) if this property holds with probability tending to 1 as the number of variables approaches infinity.}$ it fails at only one constraint. We introduce the following definitions.
\begin{definition}
Let $I$ be a CSP instance. A constraint $C$ is called a self-unsatisfiable constraint if there exists an assignment under which $C$ is the only unsatisfied constraint in $I$. If variable $x$ is contained in $C$, then $x$ is called a self-unsatisfiable variable. If $I$ is unsatisfiable and every variable is a self-unsatisfiable variable, then $I$ is called a self-unsatisfiable formula.
\end{definition}
\begin{lemma}\label{lem:selfunsat}
Let $I$ be a random CSP instance of Model RB with $n$ variables and $rn\ln d$ constraints. If $I$ is unsatisfiable, then w.h.p.  $I$ is  a self-unsatisfiable formula.
\end{lemma}
\begin{proof}
First we show that for any constraint $C$ of $I$,  with positive probability there exists an assignment which satisfies all constraints except $C$. In fact let $N$ be the number of such assignments, then
\begin{align}\label{eq:expectationN}
\mathbf{E}[N]=d^n(1-p)^{rn\ln d-1}p.
\end{align}
Using a similar argument as in \cite{xu2000}, we have
\begin{align*}
\mathbf{E}[N^2]
=&\sum_{S=0}^nd^n\binom n S (d-1)^{n-S}\left((1-p)\frac{\binom Sk}{\binom nk}+(1-p)^2(1-\frac{\binom Sk}{\binom nk})\right)^{rn\ln d-1}\cdot\\
&\left(p\frac{\binom Sk}{\binom nk}+p^2(1-\frac{\binom Sk}{\binom nk})\right).
\end{align*}
Hence, (\ref{eq:secondmm}),(\ref{eq:cauchy}) and (\ref{eq:expectationN}) ensure that
\begin{align*}
\frac{\mathbf{E}[N^2]}{\mathbf{E}[N]^2}=\frac{\mathbf{E}[X^2]}{\mathbf{E}[X]^2}\cdot\frac{1+\frac{1-p}{p}s^k}{1+\frac{p}{1-p}s^k}\le\frac1p\frac{\mathbf{E}[X^2]}{\mathbf{E}[X]^2}\le\frac3p.
\end{align*}
Then $\mathbf{Pr}(N>0)\ge\frac{\mathbf{E}[N]^2}{\mathbf{E}{[N^2]}}\ge\frac p3$. Therefore the probability that $C$ is a self-unsatisfiable constraint is at least $\frac p3$.

Next, since the number of constraints is $rn\ln d$, the average degree of each variable $x\in\mathcal{X}$ is $rk\ln d$. By the Chernoff Bound,
\begin{align*}
\mathbf{Pr}\left[\text{Deg}(x)\le\frac1{100}rk\ln d\right]\le e^{-(99/100)^2rk\ln d/2}=n^{-(99/100)^2rk\alpha/2},
\end{align*}
where  Deg$(x)$ denotes the degree of variable $x$. From the requirement (\ref{de:alpha}) we know that $1-(99/100)^2rk\alpha/2<0$, thus
\begin{align}\label{eq:degree}
n\mathbf{Pr}\left[\text{Deg}(x)\le\frac1{100}rk\ln d\right]\le o(1).
\end{align}
Therefore, almost surely all variables have degree at least $\frac1{100}rk\ln d$.

Furthermore, note that each variable appears in at least $\frac1{100}rk\ln d$ constraints and the probability that each constraint appears to be a self-unsatisfiable constraint is at least $\frac p3$, thus the probability that $x$ is not a self-unsatisfiable variable is at most $\left(1-\frac p3\right)^{\frac1{100}rk\ln d}$.

Note that (\ref{de:alpha}) entails that $1+\frac1{100}rk\alpha\ln(1-p/3)<0$, therefore the probability that there exists a variable which is not self-unsatisfiable is at most
\begin{align}\label{eq:selfformula}
n\left(1-\frac p3\right)^{\frac1{100}rk\ln d}=n^{1+\frac1{100}rk\alpha\ln(1-p/3)}=o(1).
\end{align}
Thus w.h.p. all variables are self-unsatisfiable variables.
\end{proof}
\begin{remark}
We claim that Lemmas \ref{lem:onethird} and \ref{lem:selfunsat} hold for any domain size greater than polynomial. Take exponential domain size $d=\beta^n$ (where $\beta>1$ is a constant) for example, similarly, the dominant contributions of $\mathbf{E}[X^2]/\mathbf{E}[X]^2$ come from $F(0)=1-o(1)$ and $F(n)=1/\mathbf{E}[X]$, and asymptotic calculations show that $F(i)=\Theta\left((\frac nd)^{i}\right)$ and
\begin{align*}
F(n-i)=(1+o(1))\frac1{\mathbf{E}[X]}\exp\left\{i\ln n+i\Big(1+\frac{pk}{\ln(1-p)}\Big)\ln d\right\}
\end{align*}
are negligible for small integer $i$, since $1+\frac{pk}{\ln(1-p)}<0$. Moreover, probability analysis holds more easily in the proof of Lemma \ref{lem:selfunsat} if $d=\beta^n$ ((\ref{eq:degree}) and (\ref{eq:selfformula})).

\end{remark}

Next we define a \emph{symmetry mapping} of a constraint which changes its permitted set slightly.
\begin{definition}
Consider a random instance $I$ of Model RB with $k=2$.  Assume that $C=(X,R)$ is a constraint of $I$ and $X=(x_1,x_2)$, then a symmetry mapping of $C$ is to change $R$ by choosing $u_1,u_2\in D_1$ ($u_1\ne u_2$), $v_1,v_2\in D_2$ ($v_1\ne v_2$), where $(u_1,v_1),(u_2,v_2)\in R$ and $(u_1,v_2),(u_2,v_1)\notin R$, and then exchanging $u_1$ with $u_2$ (see Figure \ref{map}).
\end{definition}
\begin{figure}[h]
  \centering
  \includegraphics[width=0.9\columnwidth]{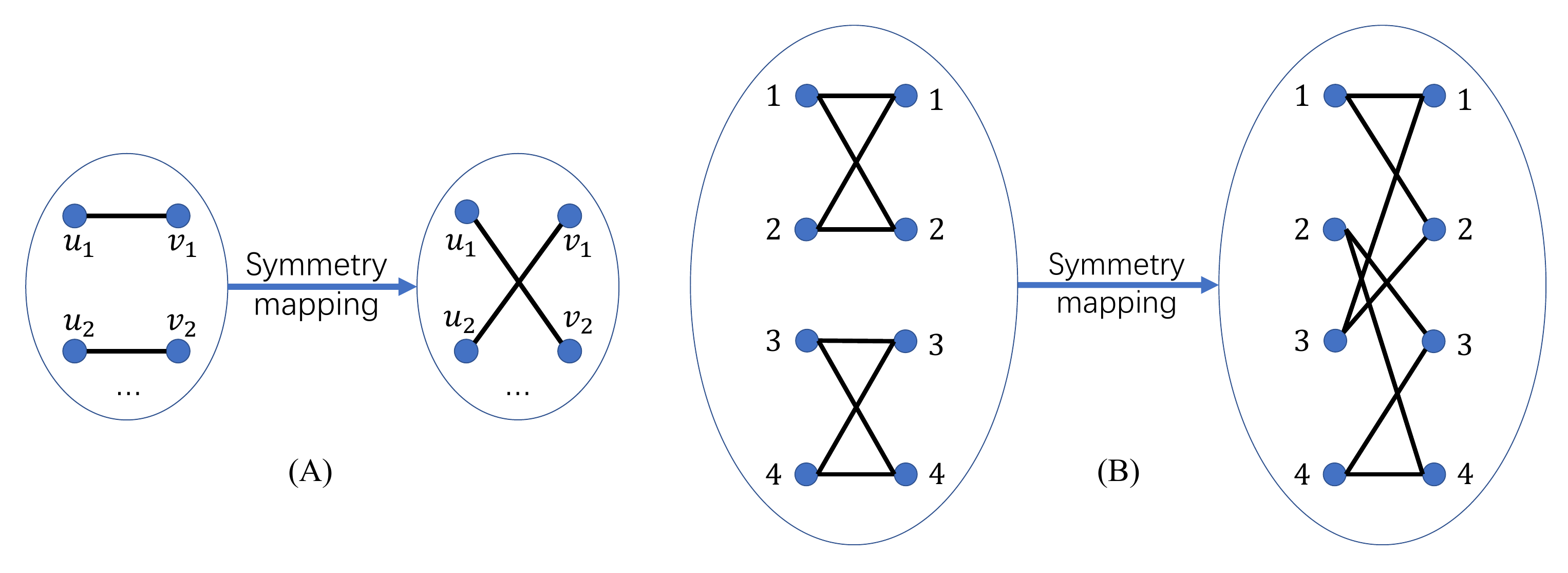}
  \caption{\footnotesize{(A) shows a symmetry mapping of the constraint $C$ by exchanging $u_1$ with $u_2$. (B) shows  an example of a symmetry mapping by exchanging the domain value $2$ with $3$ in $D_1$, where we set $D_1=D_2=\{1,2,3,4\}$ and the original permitted set $R=\{(1,1),(1,2),(2,1),(2,2),(3,3),(3,4),(4,3),(4,4)\}$. $R$ becomes $\{(1,1),(1,2),(2,3),(2,4),(3,1),(3,2),(4,3),(4,4)\}$ after such a symmetry mapping. Thus for an assignment $\sigma$, if $\sigma(x_1)=2,\sigma(x_2)=2$, then $C$ changes from satisfiable to unsatisfiable under $\sigma$, since $(2,2)$ does not belong to the permitted set $R$ any more; for an assignment $\tau$, if $\tau(x_1)=2,\tau(x_2)=3$, then $C$ changes from unsatisfiable to satisfiable under $\tau$, since $(2,3)$  belongs to the permitted set $R$ after this symmetry mapping.
}}\label{map}
\end{figure}
With the above properties of Model RB, we obtain the following interesting results.
\begin{theorem}\label{le:indistinguishable}
There exists an infinite set of satisfiable and unsatisfiable instances of Model RB such that this set is a fixed point under the symmetry mapping of changing satisfiability.

\end{theorem}
\begin{proof}
Let $\mathcal{I}$ be the set of  RB instances with $n$ variables and $rn\ln d$ constraints, where each instance either has a unique solution or has no solution.

Assume that $I\in\mathcal{I}$ has exactly one solution $\sigma$, which happens with probability at least $1/6$ from Corollary \ref{lem:solutionpair}. For an arbitrary constraint $C'$, assume without loss that $C'=(X',R'),X'=(x,x_1),\sigma(x)=u,\sigma(x_1)=v,$ and $D,D_1$ are the domains of $x$ and $x_1$, respectively. By the symmetry requirement, there exist $u'\in D,v'\in D_1$ such that $(u',v')\in R', (u,v')\notin R',(u',v)\notin R'$, then we will exchange $u$ with $u'$. It is easy to see that this symmetry mapping  will convert $(u,v)$ into an unpermitted tuple and convert $(u,v'),(u',v)$ into permitted tuples, thus $\sigma$ is no longer a solution. However, it is possible that at most $2(1-p)d$ pairs $(u,*),(u',*)\in R$ can be expanded to new solutions (this is because by the symmetry requirement, the degree of each domain value of each variable is $(1-p)d$). Specifically, the probability that $(u,v')$ can be expanded to a new solution is at most $d^{n-2}(1-p)^{rn\ln d-1}=\frac1{2(1-p)d^2}$, thus a simple calculation yields that the probability that $I$ is still satisfiable after the symmetry mapping is at most $O(\frac1{d})=o(1)$.

Assume that $I\in\mathcal{I}$ is an unsatisfiable instance, then w.h.p. all variables in $I$ are self-unsatisfiable variables from Lemma \ref{lem:selfunsat}. This implies that there exist a constraint $C''\in\mathcal{C}_{x}$ and an assignment $\tau$  such that $C''$ is the only unsatisfied one under $\tau$.  Assume without loss that $C''=(X'',R''),X''=(x,x_2),\tau(x)=u,\tau(x_2)=w,$ and $D,D_2$ are the domains of $x$ and $x_2$, respectively. It is apparent that $(u,w)\notin R''$. By our symmetry requirement, there exist $u'\in D,w'\in D_2$, where $(u,w'),(u',w)\in R''$ and $(u',w')\notin R''$, such that a symmetry mapping of exchanging $u$ with $u'$ will  convert $(u,w)$ into a permitted tuple, thus $C''$ becomes satisfiable under $\tau$. Moreover, using a similar argument as above, we can see that the probability that the new pairs $(u,*),(u',*)\in R$ could expand to solutions is at most $O(\frac1{d})=o(1)$. Thus w.h.p. $I$ has only one solution after this symmetry mapping.

From the above two cases we can see that for any $I\in\mathcal{I}$, the symmetry mapping changes its satisfiability, however, $I$ still belongs to $\mathcal{I}$ after the mapping, thus $\mathcal{I}$ can be considered as a fixed point under the symmetry mapping.
\end{proof}

In this section, it has been shown that we can construct satisfiable and unsatisfiable instances using exactly the same method and the same parameter values. Moreover, these satisfiable and
unsatisfiable instances can be transformed into each other by performing the same mapping. This property is very similar to that of the self-referential proposition introduced by Kurt G\"{o}del  \cite{godel1931}
in order to prove that such a proposition can be neither proved nor disproved (i.e. whether this proposition is true or false cannot be distinguished in finite formal systems).
G\"{o}del's results reveal the fundamental difference  between the syntax defined by rules and the semantics defined by models. An algorithm is a finite sequence of rules, which can also be
viewed as a finite formal system. On the other hand, the exhaustive search method is used to determine, based on the semantic definition of a property, whether this property is satisfied by a model (an assignment).
 Inspired by G\"{o}del's idea, we call these satisfiable and unsatisfiable instances the most indistinguishable examples because the self-referential property makes such examples syntactically hard to be distinguished from each other.
\section{Main Results}
Proving complexity lower bounds (algorithmic impossibility results) for a given problem is essentially reasoning and making conclusions about an infinite set of algorithms. In mathematics, any such conclusion should be based on assumptions about the nature of the infinite set. These assumptions must be consistent with the reality and usually appear as axioms. Similar to many combinatorial problems, the general CSP has no global structure that can be exploited to design algorithms. Currently, the only exact algorithm available for solving CSP is a divide-and-conquer algorithm systematically searching the solution space with various pruning strategies. 

In this paper, we view an algorithm as a finite formal system which is defined by a finite set of symbols and inference rules. It is well known that finite formal systems are stronger than algorithms (Turing machines) because algorithms must be finite formal systems but finite formal systems are not necessarily algorithms. Here we use finite formal systems to solve CSP, i.e. determining if a CSP instance is satisfiable. We only need to assume that this task is finished by dividing the original problem into subproblems. Under this assumption, we have the following lemma.

\begin{lemma}\label{lm:subproblem}
If a CSP problem with $n$ variables and domain size $d$ can be solved in $T(n)=O(d^{cn})$ time ($0<c<1$ is a constant), then at most $O(d^c)$ subproblems with $n-1$ variables are needed to solve the original problem.
\end{lemma}

\begin{proof}
It is easy to see that $d$ subproblems with $n-1$ variables are sufficient to solve the original problem. By condition, a CSP problem with $n-1$ variables can be solved in $T(n-1)$ time. Note that
$$T(n)=O(d^{cn})=O(d^cd^{c(n-1)})=O(d^c)T(n-1),$$
 thus at most $O(d^c)$ subproblems with $n-1$ variables are needed to solve the original problem.
\end{proof}
The above lemma is a key to proving the main results of this paper. For a better understanding, we take the classical sorting problem as an example. Assume that our goal is to prove that sorting $n$ elements cannot be done in $T(n)=O(n)$ time. By condition, we have $T(n)=O(n-1)+O(1)=T(n-1)+O(1)$. This means that we need a subproblem of $n-1$ elements with additional $O(1)$ operations to solve the original problem. We can show by contradiction that this is impossible and so finish the proof.

\begin{theorem}\label{th:main}
Model RB  cannot be solved in  $O(d^{c n})$ time for any constant $0<c<1$.
\end{theorem}
\begin{proof}
For the sake of simplicity, we will  prove that Theorem \ref{th:main} holds  for Model RB with  $k=2$ by contradiction. Let $I$ be a RB instance with $n$ variables and $rn\ln d$ constraints.
Suppose there exists some constant $0<c<1$ such that  $I$ can be solved in $O(d^{c n})$ time, then Lemma \ref{lm:subproblem} implies that $I$ can also be solved by assigning at most $O(d^c)$ values to an arbitrary variable, say $x$, and then solving the resulting subproblems (with $n-1$ variables)  which require $O(d^{c (n-1)})$ time. We will show that there exist instances where the $O(d^c)$ subproblems produced by assigning $O(d^c)$ values to $x$ are impossible to determine its satisfiability. For convenience, let $D$ be the domain of $x$, $\widetilde{D}$ be the set of $O(d^c)$ values which have been assigned to $x$ ($|\widetilde{D}|=O(d^c)$), and $\mathcal{C}_{x}$ be the set of constraints containing $x$.

Follow the strategy of the proof of Theorem \ref{le:indistinguishable}, if $I$ is an instance having exactly one solution $\sigma$, then the constraint $C'$ will be  arbitrarily selected from $\mathcal{C}_{x}$.  Note that $O(d^c)$  is $o(1)$ compared with the domain size $d$, thus the probability that $u$ belongs to $\widetilde{D}$ is $o(1)$. Therefore the symmetry mapping in Theorem \ref{le:indistinguishable} will be performed by exchanging $u$ with $u'$, where $u'$ is chosen from $D\backslash \widetilde{D}$, then $I$ becomes unsatisfiable.
Similarly, if $I$ is an unsatisfiable instance, then the symmetry mapping of exchanging $u$ with $u'$, where $u'$ is chosen from $D\backslash \widetilde{D}$, will  convert $(u,w)$ into a permitted tuple, thus $I$ becomes satisfiable.

Note that in either of the above cases, the $O(d^c)$ subproblems with $n-1$ variables remain unchanged, while the satisfiability of the original problem with $n$ variables has been changed after performing the symmetry mapping. We can conclude that the $O(d^c)$ subproblems are insufficient thus impossible to determine if $I$ is satisfiable or unsatisfiable. This completes the proof of Theorem \ref{th:main}.
\end{proof}

As mentioned before, finite formal systems are stronger than algorithms (Turing machines). The above theorem indicates that it is impossible to find a finite formal system for solving Model RB that is substantially better than (and thus can replace) the exhaustive search method based on the semantic definition of satisfiability. More interestingly, the above proof follows the same method (i.e. diagonalization and self-reference) used by Kurt G\"{o}del \cite{godel1931} and Alan Turing \cite{turing1936}, respectively, in their epoch-making papers of proving logical and computational impossibility results. This method is very simple, elegant and powerful. On the other hand, the nice mathematical properties and the extreme hardness of Model RB are the key to making it feasible and effective for proving algorithmic impossibility results. We hope that the basic idea of this paper (i.e. Constructing the Most Indistinguishable Examples) will open the way for proving complexity lower bounds that has always been a challenging task even for many polynomial-time solvable problems.

It is worth mentioning that Xu and Li \cite{xu2006} established a direct link between proved phase transitions and the abundance with hard examples by proving that Model RB has both of these two properties. They also made a detailed comparison between Model RB and some other well-studied models with phase transitions, and stated that ``\emph{such mathematical tractability should be another advantage of RB/RD, making it possible to obtain some interesting results which do not hold or cannot be easily obtained for random 3-SAT}". This paper shows that as expected, interesting results can be achieved in an uncomplicated manner. More than two thousand years ago, Laozi, a great Chinese thinker and philosopher, once said: ``\emph{The greatest truths are the simplest}". We believe that the truth of computational hardness should also follow this basic and universal principle. So we can say that the results of this paper are surprising, but not hard to understand.

CSP can be encoded into SAT by use of the log-encoding \cite{walsh2000} which introduces a new Boolean variable for each bit in the domain value of each CSP variable and thus uses a logarithmic number of Boolean variables to encode domains.  It must be emphasized that each clause of these encoded SAT instances could be very long with $\Theta (\log_{2}{d})$ Boolean variables if the domain size $d$ grows with the number of CSP variables $n$.  This is in sharp contrast to 3-SAT that is very short in clause length and has received the most attention in the past half-century.
For encoded SAT instances of Model RB using the log-encoding, we have totally $N=n\log_{2}{d}$ Boolean variables, $O(Nd^k)$ clauses and $O(Nd^k\log_{2}{d})$ literals. Note that $d^{cn}=2^{cn\log_2d}=2^{cN}$, so it is easy to derive the following corollary from Theorem \ref{th:main}.
\begin{corollary}
SAT with $N$ Boolean variables cannot be solved in $O(2^{cN})$ time for any constant $0<c<1$.
\end{corollary}

The above corollary holds for SAT with no restriction on the clause length, and so is not directly applicable to $k$-SAT with constant $k\ge 3$ whose lower bounds can be obtained by reduction from encoded SAT instances of Model RB. Other CSP instances of domain size $l$ can also be encoded from Model RB using $N=n\log_{l}{d}$ variables, thus cannot be solved in $O(l^{cN})$ time for any constant $0<c<1$.

It is well-known that SAT is NP-complete \cite{cook1971}, and so it follows that P $\ne$ NP.

\newpage


\end{document}